\documentclass[letterpaper,11pt]{article}
\usepackage[T1]{fontenc}
\usepackage[english]{babel}
\usepackage{amsmath,amsfonts,amsthm,amssymb,color}
\usepackage{fancybox}
\usepackage{hyperref}

\usepackage{thm-restate}

\usepackage[letterpaper,margin=1in]{geometry}

\newtheorem{theorem}{Theorem}[section]
\newtheorem{corollary}[theorem]{Corollary}
\newtheorem{lemma}[theorem]{Lemma}

\newtheorem{fact}[theorem]{Fact}

\theoremstyle{definition}
\newtheorem{definition}[theorem]{Definition}

\newenvironment{fminipage}%
  {\begin{Sbox}\begin{minipage}}%
  {\end{minipage}\end{Sbox}\fbox{\TheSbox}}

\newenvironment{algbox}[0]{\vskip 0.2in
\noindent 
\begin{fminipage}{6.3in}
}{
\end{fminipage}
\vskip 0.2in
}

\def\prob#1#2{\mbox{Pr}_{#1}\left[ #2 \right]}

\def\expec#1#2{{\mathbb{E}}_{#1}\left[ #2 \right]}

\def\abs#1{\left|#1  \right|}

\def\norm#1{\left\| #1 \right\|}

\newcommand\numberthis{\addtocounter{equation}{1}\tag{\theequation}}

\newcommand\pphi{\boldsymbol{\mathit{\phi}}}

\newcommand\PPi{\boldsymbol{\Pi}}

\newcommand\cchi{\boldsymbol{\chi}}

\newcommand\boldone{\boldsymbol{1}}

\newcommand\poly{\mathrm{poly}}

\newcommand\dd{\boldsymbol{\mathit{d}}}

\newcommand\ff{\boldsymbol{\mathit{f}}}

\newcommand\ww{\boldsymbol{\mathit{w}}}

\renewcommand\AA{\textbf{A}}
\newcommand\BB{\textbf{B}}

\newcommand\DD{\textbf{D}}

\newcommand\II{\mathbf{I}}

\newcommand\LL{\mathbf{L}}

\newcommand\XX{\mathbf{X}}

\newcommand\SC{\boldsymbol{\textsc{SC}}}

\newcommand{\vect}[1]{\ensuremath{\mathbf{#1}}}

\newcommand\Htil{\tilde{H}}
\newcommand\Otil{\tilde{O}}

\newcommand\ehat{{\hat{{e}}}}
\newcommand\mhat{{\hat{{m}}}}
\newcommand\nhat{{\hat{{n}}}}

\newcommand\uhat{{\hat{{u}}}}
\newcommand\vhat{{\hat{{v}}}}

\newcommand\er{\mathcal{R}_{\text{eff}}}

\begin{document}

\title{Fully Dynamic Effective Resistances
}

\author{
David Durfee
\footnote{
Emails: \texttt{\{ddurfee,ygao380\}@gatech.edu},
~\texttt{gramoz.goranci@univie.ac.at},
~\texttt{rpeng@cc.gatech.edu}.}
\footnote{This material is based upon work supported by the
National Science Foundation under Grant No. 1718533.}
\\
Georgia Tech
\and
Yu Gao\footnotemark[2]
\\
Georgia Tech
\and
Gramoz Goranci
\footnotemark[1]
\footnote{This work was done while visiting the
Georgia Institute of Technology. The research leading
to these results has received funding from the Marshall Plan Foundation and the European Research Council under the European
Union's Seventh Framework Programme (FP/2007-2013) / ERC Grant Agreement no. 340506.}
\\
University of Vienna
\and
Richard Peng\footnotemark[1]~\footnotemark[2]
\\
Georgia Tech
}

\pagenumbering{gobble}

\maketitle

\begin{abstract}
In this paper we consider the \emph{fully-dynamic} All-Pairs Effective Resistance problem, where the goal is to maintain effective resistances
on a graph $G$  among any pair of query vertices under an intermixed
sequence of edge insertions and deletions in $G$.
The effective resistance between a pair of vertices is a physics-motivated quantity that encapsulates both the congestion and the dilation of a flow.
It is directly related to random walks, and it has been instrumental in the recent works for designing fast algorithms for combinatorial optimization problems, graph sparsification, and network science.

We give a data-structure that maintains $(1 \pm \epsilon)$-approximations to all-pair effective resistances of a fully-dynamic unweighted, undirected multi-graph $G$ with $\tilde{O}(m^{4/5}\epsilon^{-4})$ expected amortized update and query time, against an oblivious adversary.
Key to our result is the maintenance of a dynamic
\emph{Schur complement}~(also known as vertex resistance sparsifier)
onto a set of terminal vertices of our choice.

This maintenance is obtained (1) by interpreting the Schur complement
as a sum of random walks and (2) by randomly picking the vertex subset into which  the sparsifier is constructed. We can then show that each update in the graph affects a small number of such walks, which in turn leads to our sub-linear update time. We believe that this local representation of vertex sparsifiers may be of independent interest.

\end{abstract}


\newpage

\pagenumbering{arabic}

\section{Introduction}
\label{sec:Introduction}

The incorporation of numerical and optimization tools into graph
algorithms has been a highly successful approach in algorithm design.
It is key to the current best results for several fundamental
problems in combinatorial optimization, such as
approximate maximum flow~\cite{ChristianoKMST11,Sherman13,KelnerLOS14,Peng16},
multi-commodity flow~\cite{Sherman17},
shortest path and weighted matching~\cite{CohenMSV17}.
For dynamic graphs undergoing edge modifications,
core components from these algorithms such as graph partitioning,
have also played an important role in recent
developments on dynamic minimum spanning forest with worst-case guarantees~\cite{Wulffnilsen17,NanongkaiSW17,NanongkaiS17}.
The versatility of these tools in the static setting suggests that they can extend to wider ranges of problems on dynamic graphs.

Dynamic graph algorithms seek to maintain solutions to certain problems
on graphs as they undergo edge insertions and deletions in time faster
than recomputing the solution from scratch after each update.
It is a subject that has been studied extensively in data structures,
with problems being maintained including
$2-$ or $3-$connectivity~\cite{EppsteinGIN97,HolmDT01,HolmRW15},
shortest paths~\cite{HenzingerKN14,HenzingerKN16,BernsteinC16,AbrahamCK17},
global minimum cut~\cite{Henzinger97,Thorup07,LackiS11,GoranciHT16},
maximum matching~\cite{OnakR10,GuptaP13,BhattacharyaHN16},
and maximal matching~\cite{Baswana15,NeimanS16,Solomon16}.
On general graphs, these results give sub-linear time only under significant
restrictions to the queries, e.g.,
global minimum cut, $3$-edge connectivity,
size of maximum matching (which equals to the value of
a particular $s-t$ max-flow), and
shortest paths from a fixed source.
Relaxing such restrictions, specifically maintaining solutions to
more general problems on dynamic graphs, is a major area of
ongoing research in data structures.

A more unified view of these problems is through the optimization 
or the numerical algorithms perspective.
In the undirected setting, where much work on dynamic graph algorithms
has taken place, the maximum flow problem and the shortest path problem correspond to minimizing the $\ell_{\infty}$ and the $\ell_{1}$ objectives of flows meeting a certain demand, respectively.
A natural interpolation is the $\ell_2$ objective, otherwise known as the electrical flow problem. In the static setting, $\ell_2$ primitives have been instrumental
in recent works on faster graph 
algorithms~\cite{ChristianoKMST11,Madry13,ChinMMP13, Madry16,AllenZhuLOW17,CohenMTV17}.
The objective value of this flow, effective resistances, is also a well
studied quantity, and has direct applications in
random walks, spanning trees~\cite{MadryST15}, and importance
of graph edges~\cite{SpielmanS08:journal}.
An in-depth treatment of such connections can be found in the monograph
by Doyle and Snell~\cite{DoyleS84}, and a short summary of the connection
between $\ell_p$ norm flows and resistances is given in Appendix~\ref{sec:Flows}.

%

The applications of effective resistances,
and its broader connections with optimization problems on graphs
make it arguably one of the most important primitives in algorithm design.
However, in the dynamic setting, prior to our work, sub-linear time
algorithms for maintaining effective resistances were only known in minor-free
graphs~\cite{GoranciHP17,GoranciHP18:arxiv}.

\paragraph*{Our result.} 

In this paper, we show the first sub-linear time
fully-dynamic algorithm for maintaining approximate
effective resistances in general graphs.
Our algorithm is randomized and the approximation guarantee holds
with probability at least $1-1/\poly(n)$.


\begin{theorem}
\label{thm:Main}
For any given error threshold $\epsilon > 0$,
there is a data-structure for maintaining an unweighted, undirected multi-graph $G=(V,E)$ with up to $m$ edges that supports the following operations
in $\tilde{O}(m^{4/5}\epsilon^{-4})$ expected amortized time:
\begin{itemize}
	\item \textsc{Insert}$(u,v)$: Insert the edge $(u,v)$ with resistance $1$ in $G$.
	\item \textsc{Delete}$(u,v)$: Delete the edge $(u,v)$ from $G$.
	\item \textsc{EffectiveResistance}$(s,t)$: Return a $(1 \pm \epsilon)$-approximation to the effective resistance between $s$ and $t$ in the current graph $G$. 
\end{itemize}
\end{theorem}



If we restrict to \emph{simple} graphs,
and only maintain the effective resistances between a small number
of vertex-pairs $(s_{i},t_{i})$,
our algorithm can also give a running time of $\Otil(n^{6/7}\epsilon^{-4})$
per operation, which is truly sub-linear irrespective of graph density.
 We discuss such an improvement in Section~\ref{sec:DependOnN}.

Our algorithm is motivated by two sequences of recent results.
For special family of graphs, e.g. planar graphs, recent works~\cite{GoranciHP17,GoranciHP18:arxiv} showed data structures for maintaining approximate effective resistances with $\Otil(\sqrt{n}\epsilon^{-2})$
\footnote{We use the notation of $\Otil(f(n))$ to hide polylog factors.
Specifically, for a function $f(n)$, some error parameter $\epsilon$, and some  constant $c > 0$, we have $\Otil(f(n)) = O(f(n) \log^{c}f(n) \log^{c} (1/\epsilon))$.} update and query time.
On the other hand, concurrent results on generating random spanning trees~\cite{DurfeeKPRS17} and computing network centrality parameters~\cite{LiZ18} implicitly rely on the ability to answer effective resistance queries along with edge deletions or contractions offline.

The key algorithmic component behind our result is the efficient maintenance of a graph reduction to a random subset of vertices,
picked uniformly from the endpoints of edges.
To this end, we leverage recent results based on vertex
elimination schemes for solving linear systems in
graph Laplacians~\cite{KyngLPSS16,KyngS16}
and generating random spanning trees~\cite{DurfeePPR17}.
Specifically, we combine random walk based methods for
generating effective resistance preservers on terminals
with results in combinatorics that bound the speed at which
such walks spread among vertices~\cite{BarnesF96}.
By showing that these walks are sufficiently local, we conclude
that each update only affects a small number of such walks, and those can in turn be resampled efficiently. 
To the best of our knowledge, we are the first to utilize the behavior
of random walks in data structures for maintaining graphs that undergo edge insertions and deletions.

Finally, we remark that the effective resistance between a pair of vertices
is finite only if they are connected.
Thus, our data structure also provides the first scheme
for maintaining connectivity in dynamic graphs that does not utilize
data structures for maintaining trees under edge updates,
such as link-cut trees~\cite{SleatorT83,SleatorT85},
Euler-tour trees~\cite{HenzingerK95},
or Top-trees~\cite{AlstrupHLT05}.
The many extensions of ideas stemming from dynamic connectivity,
and the wide range of applications of random walks make us optimistic that our
algorithmic ideas could be useful for dynamically maintaining other
important graph properties.

\subsection{Related Works}
\label{subsec:related}

The recent data structures for maintaining effective resistances in planar
graphs~\cite{GoranciHP17,GoranciHP18:arxiv} drew direct connections between
Schur complements and data structures for maintaining them
in dynamic graphs.
This connection is due to the  preservation of effective resistances
under vertex eliminations (Fact~\ref{fact:SchurER}).
From this perspective, the Schur complement can be viewed as a
vertex sparsifier for preserving resistances among a set of terminal vertices.

The power of vertex or edge graph sparsifiers,
which preserve certain properties while reducing problem sizes,
has long been studied in data structures~\cite{Eppstein91,EppsteinGIN97}.
Ideas from these results are central to recent works on offline
maintenance for $3$-connectivity~\cite{PSS17:arxiv}, generating
random spanning trees~\cite{DurfeeKPRS17}, and new notions of centrality for networks~\cite{LiZ18}.
Our result is the first to maintain such vertex sparsifiers,
specifically Schur complements, for \emph{general} graphs in online settings.

While the ultimate goal is to dynamically maintain (approximate)
minimum cuts and maximum flows, 
effective resistances represent a natural `first candidate' for this
direction of work due to them having perfect vertex sparsifiers.
That is, for any subset of terminals, there is a sparse graph on them
that approximately preserves the effective resistances among all
pairs of terminals.
This is in contrast to distances, where it's not known whether
such a graph can be made sparse, or in contrast to cuts, where the existence of
such a dense graph is not known.


\paragraph*{Dynamic Graph Algorithms.}

The maintenance of properties related to paths in dynamic graphs
is a well studied topic on data structures~\cite{Frederickson85,
EppsteinGIN97,HolmDT01,KapronKM13,Wulffnilsen17,NanongkaiSW17,NanongkaiS17}.
A key difficulty facing paths on graphs is that general graphs are not
decomposable: piecing together connectivity information from an arbitrary
partition of a graph is difficult, and there are classes of graphs such as expanders that are not partitionable.

Dynamic algorithms for evaluating algebraic functions such as matrix determinant and matrix inverse has also been considered~\cite{Sankowski04}. One application of such algorithms is that they can be used to dynamically maintain single-pair effective resistance. Specifically, using the dynamic matrix inversion algorithm, one can dynamically maintain \emph{exact} $(s,t)$-effective resistance in $O(n^{1.575})$ update time and $O(n^{0.575})$ query time.


\paragraph*{Vertex Sparsifiers.}

Vertex sparsifiers have been studied in more general settings for
preserving cuts and flows among terminal
vertices~\cite{Moitra09,CharikarLLM10,KrauthgamerR13}.
Efficient versions of such routines have direct applications in
data structures, even when they only work in restricted
settings: terminal sparsifiers on quasi-bipartite graphs~\cite{AndoniGK14}
were core routines in the data structure for
maintaining flows in bipartite undirected graphs~\cite{AbrahamDKKP16}.

Our data structure utilizes vertex sparsifiers, but in even more
limited settings as we get to control the set of vertices to sparsify onto.
Specifically, the local maintenance of this sparsifier under insertions
and deletions hinges upon the choice of a random subset of terminals,
while vertex sparsifiers usually need to work for any subset of terminals.
Evidence from numerical algorithms~\cite{KyngLPSS16,DurfeePPR17} suggest
this choice can significantly simplify interactions between algorithmic
components.
We hope this flexibility can motivate further studies of vertex sparsifiers
in more restrictive, but still algorithmically useful settings.

\paragraph*{Organization. }

The paper is organized as follows.
We discuss preliminaries in Section~\ref{sec:Preliminaries},
after which we present our data-structure in Section~\ref{sec:Overview}.
The key properties of random walks that we use are given in Section~\ref{sec:propertiesRandomWalk}, and
we show the dynamic maintenance of approximate Schur complements
in Section~\ref{sec:Dynamic}.
In Section~\ref{sec:DependOnN}, we briefly discuss an alternate
way of analyzing our data structure's performances in terms of the
number of vertices.
In Appendix~\ref{sec:SchurComplement}, we provide details on the
graph approximation guarantees that our random walk sampling routines
rely on.
Finally, in Appendix~\ref{sec:Flows} we provide brief details on
the $p-$norm flow formulations of shortest paths, maximum flows,
and effective resistances.

\section{Preliminaries}
\label{sec:Preliminaries}

In our dynamic setting, an undirected, unweighted multi-graph undergoes both insertions and deletions of edges. We let $G = (V, E)$ always refer to the \emph{current} version of the graph.  
We will use $n$ and $m$  to denote bounds on the number
of vertices and edges at any point, respectively.

A \emph{walk} in $G$ is a sequence of vertices such that
consecutive vertices are connected by edges.
A \emph{random walk} in $G$ is a walk that starts at
a starting vertex $v_0$, and at step $i \geq 1$, the vertex $v_{i}$
is chosen randomly among the neighbors of $v_{i - 1}$.


For an unweighted, undirected multi-graph $G$, let $\AA_G$ denote its adjacency matrix and let $\DD_G$ its degree diagonal matrix~(counting edge multiplicities for both matrices). The graph \emph{Laplacian} $\LL_{G}$ of $G$ is then defined as $\LL_G = \DD_G-\AA_G$. Let $\LL_G^{\dag}$ denote the Moore-Penrose pseudo-inverse of $\LL_{G}$. We also need to define the indicator vector $\boldone_{u} \in \mathbb{R}^{V}$ of a vertex $u$ such that $\boldone_u(v) = 1$ if $v = u$, and $\boldone_u(v) = 0$ otherwise. 

\paragraph*{Effective Resistance}

For our algorithm, it will be useful to define effective resistance using linear algebraic structures. Specifically, given any two vertices $u$ and $v$ in $G$, if $\cchi(u,v) := \boldone_u - \boldone_v$, then the \emph{effective resistance} between $u$ and $v$ is given by
\[
\er^{G}\left(u, v \right)
:=
\cchi_{u, v}^{T}
\LL_{G}^{\dag}
\cchi_{u, v}.
\]

Linear systems in graph Laplacian matrices can be solved in
nearly-linear time~\cite{SpielmanTengSolver:journal}.
One prominent application of these solvers is the approximation of
effective resistances.

\begin{lemma} \label{lemm:efficientEffectiveResistance}
Fix $\epsilon \in (0,1)$  and let $G=(V,E)$ be any graph with two arbitrary distinguished vertices $u$ and $v$. There is an algorithm that computes a value $\phi$ such that
\[
	\normalfont (1-\epsilon)\er^G(u,v) \leq \phi \leq (1+\epsilon)\er^G(u,v),
\]
in $\tilde{O}(m \epsilon^{-2})$ time with high probability.
\end{lemma}





\paragraph*{Schur complement.}

Given a graph $G=(V,E)$, we can think of the \emph{Schur complement} as the partially eliminated state of $G$. This relies on some partitioning of $V$ into two disjoint subset of vertices $T$ and $F$, which in turn partition the Laplacian $\LL_G$ into $4$ blocks:
\[
\LL
:=
\left[
\begin{array}{cc}
\LL_{\left[F, F\right]}
&
\LL_{\left[F, T\right]}\\
\LL_{\left[T, F\right]}
&
\LL_{\left[T, T\right]}
\end{array}
\right].
\]

The Schur complement onto $T$, denoted by
$\textsc{SC}(G, T)$ is the matrix after eliminating
the variables in $F$.
Its closed form is given by
\[
\SC\left(G, T \right)
=
\LL_{\left[T, T\right]}
-
\LL_{\left[T, F\right]}
\LL_{\left[F, F\right]}^{\dag}
\LL_{\left[F, T\right]}.
\]

It is well known that $\SC(G,T)$ is a Laplacian matrix of a graph on vertices in $T$. To simplify our exposition, we let $\SC(G,T)$ denote both the Laplacian and its corresponding graph. 

In this work, we will not utilize the above algebraic expression of Schur complement. Instead, our algorithm is built upon a view of the Schur complement as a collection of random walks. This particular view we be discussed in more details in Section~\ref{sec:Overview}.

The key role of Schur complements in our algorithms
stems from the fact that they can be viewed as vertex sparsifiers that preserve pairwise effective resistances~(see e.g.,~\cite{GoranciHP17}).
\begin{fact}[Vertex Resistance Sparsifier]
\label{fact:SchurER}
For any graph $G=(V,E)$, any subset of vertices $T$,
and any pair of vertices $u, v \in T$,
\[ \normalfont
\er^{G}\left(u, v \right)
=
\er^{\textsc{SC}\left(G, T \right)}
\left(u, v\right).
\]
\end{fact}

\paragraph*{Spectral Approximation} 

\begin{definition}[Spectral Sparsifier] \label{def: specSpar} Given a graph $G=(V,E)$ and $\epsilon \in (0,1)$, we say that a graph $H=(V,E')$ is a $(1 \pm \epsilon)$-\emph{spectral sparsifier} of $G$ (abbr. $H \approx_{\epsilon} G$) if $E' \subseteq E$, and for all $\vect{x} \in \mathbb{R}^{n}$ 
	\[ (1-\varepsilon)\vect{x}^{T}\LL_{G}\vect{x} \leq \vect{x}^{T}{\LL_{H}}\vect{x} \leq (1+\varepsilon)\vect{x}^{T}\LL_{G}\vect{x}. \]
\end{definition}

In the dynamic setting, Abraham et al.~\cite{AbrahamDKKP16} recently
showed that $(1\pm\epsilon)$-spectral sparsifiers of a dynamic graph $G$
can be maintained efficiently.

\begin{lemma}[\cite{AbrahamDKKP16},~Theorem 4.1]
\footnote{Version 1~\url{https://arxiv.org/pdf/1604.02094v1.pdf}.}
\label{lem:DynamicSpectralSparsifier}
Given a graph $G$ with polynomially bounded edge weights, with high probability, we can dynamically maintain a $(1 \pm \epsilon)$-spectral sparsifier of size $\Otil(n \epsilon^{-2})$ of $G$ in $O(\log^{9} n \epsilon^{-2})$ expected amortized time per edge insertion or deletion. The running time guarantees hold against an oblivious adversary.
\end{lemma}

The above result is useful because matrix approximations also
preserve approximations of their quadratic forms. As a consequence of this fact, we get the following lemma.

\begin{lemma} \label{lem:ApproxER} If $H$ is a $(1 \pm \epsilon)$-spectral sparsifier of $G$, then for any pair of vertices $u$ and $v$
	\[
	(1-\varepsilon)R_G(u,v) \leq R_H(u,v) \leq (1+\varepsilon) R_G(u,v).
	\]
\end{lemma}




\section{Overview and Data Structure}
\label{sec:Overview}

In this section we start by discussing the key high level invariants that we maintain throughout our data structure. We then continue by describing how to use these invariants to dynamically maintain approximate Schur complements. Finally, we show that this leads to an algorithm for maintaining effective resistance, and thus proves our main result in Theorem~\ref{thm:Main}.



We now review two natural attempts for addressing our problem.
\begin{itemize}
\item First, since spectral sparsifiers preserve effective resistances (Lemma~\ref{lem:ApproxER}), we could dynamically maintain a spectral sparsifier~(Lemma~\ref{lem:DynamicSpectralSparsifier}), and then compute the $(s,t)$ effective resistance on this sparsifier. This leads to a data structure with $\poly(\log n, \epsilon^{-1})$ update time
and $\Otil(n \epsilon^{-2})$ query time.
\item Second, by the preservation of effective resistances under
Schur complements (Fact~\ref{fact:SchurER}), we could also utilize Schur complements to obtain a faster query time among a set of $\beta m$
terminals, $T$, for some reduction factor $\beta \in (0,1)$,
at the expense of a slower update time.
Specifically, after each edge update, we recompute an approximate Schur complement of the sparsifier onto $T$ in
time~$\Otil(m \epsilon^{-2})$~\cite{DurfeeKPRS17},
after which each query takes $\Otil(\beta m \epsilon^{-2})$ time.
\end{itemize}

The first approach obtains sublinear update time, while the second
approach gives sublinear query time. Our algorithm stems from combining these two methods,
with the key additional observation being that adding more vertices
to $T$ makes the Schur complement algorithm more local.

The running time bottleneck then becomes computing and maintaining
$\SC(G, T)$ under edge updates to $G$.
To speed up this process, we take a more local interpretation of
a sparsifier of $\textsc{SC}(G, T)$ as a collection of random walks, each
starting at an edge of $G$ and terminating in $T$.
In Appendix~\ref{sec:SchurComplement},
we review the following result, which is implicit in previous works on block
elimination based algorithms for estimating determinants~\cite{DurfeePPR17}.

\begin{restatable}{theorem}{SparsifySchur}
\label{thm:SparsifySchur}
Let $G=(V,E)$ be an undirected, unweighted multi-graph with a subset of vertices $T$. Furthermore, let $\epsilon \in (0,1)$, and let $\rho$ be some parameter related to the concentration of sampling given by
\[
\rho = O\left( \log{n}  \epsilon^{-2} \right).
\]
Let $H$ be an initially empty graph, and for every edge $e=(u,v)$ of $G$,
repeat $\rho$ times the following procedure:
\begin{enumerate}
\item Simulate a random walk starting from $u$ until
it hits $T$ at vertex $t_1$,
\item Simulate a random walk starting from $v$ until
it hits $T$ at vertex $t_2$,
\item Let the total length of this combined walk (including edge $e$)
be $\ell$.
Add the edge $(k_1, k_2)$ to $H$ with weight
\[
\frac{1}{\rho \ell}.
\]
\end{enumerate}
The resulting graph $H$ satisfies
$\normalfont \LL_H \approx_{\epsilon} \SC(G,T)$
with high probability.
\end{restatable}

The output approximate Schur complement $H$ onto $T$ has up to
$\rho m = \tilde{O}(m\epsilon^{-2})$ edges~(that is, $T$ for each edge in $G$).
However, as we will show, $H$ does not change too much upon inserting
or deleting an edge in $G$.
Therefore, we can maintain these changes using a dynamic spectral 
sparsifier $\Htil$ of $H$, and whenever a query comes, we answer 
it on $\Htil$ in $\Otil(\abs{T} \epsilon^{-2}) = \Otil(\beta m \epsilon^{-2})$ time.

Thus the bulk of our effort is devoted to generating and maintaining
the random walks described in Theorem~\ref{thm:SparsifySchur}.
Specifically, upon insertion or deletion of an edge $e = (u,v)$ in $G$,
we only need to regenerate walks that pass through $u$ or $v$.
The cost of this depends on both the length of a walk,
as well as the maximum number of walks that passes
through a vertex $u$ (which we will refer to as the \emph{load} of $u$).
For $T$ picked arbitrarily, e.g., the leftmost $n / 2$ vertices
of a length $n$ path, both of these parameters can be large:
a walk needs about $n^2$ steps to move across the path,
and the load at the middle vertices is $\Theta(n)$.

To shorten these random walks, we augment $T$ with a random subset of vertices.
Coming back to the path example,
$\beta n$ uniformly random vertices will be roughly
$\beta^{-1}$ apart, and random walks will reach one of these $\beta n$ vertices
in about $\beta^{-2}$ steps.
Because $G$ could be a multi-graph, and we want to support queries
involving any vertex, we pick $T$ as the end points of a uniform
subset of edges.
A case that illustrates the necessity of this choice is 
a path except one edge has $n$ parallel edges.
In this case it takes $\Theta(n)$ steps in expectation for
a random walk to move away from the end points of that edge.
This choice of $T$ completes the definition of our data structure,
which we summarize in Figure~\ref{fig:GlobalVariables}, and will discuss
throughout the rest of this overview.
A variant based on sampling vertices that obtains truly sublinear time
per operation, but has more limitations on operations,
is in Section~\ref{sec:DependOnN}.

\begin{figure}

\begin{algbox}

\begin{enumerate}

\item A subset of terminal vertices $T$, obtained by including the endpoints of each edge independently, with probability at least $m^{-1/5}$.

\item A sampling overhead $\rho = O(\log n \varepsilon^{-2})$ (chosen according to Theorem~\ref{thm:SparsifySchur}).

\item Graph $H$ created by
repeating the following procedure
for each edge $e = (u,v)$,

\begin{enumerate}
\item For $i=1,\ldots,\rho$,

\begin{enumerate}
\item Generate random walks $W(e, i)$
from $u$ and $v$ until either $O(m^{2/5}\log^3{n})$ steps
have been taken, or they reach $T$.

\item If both walks reach $T$ at $t_1$ and $t_2$ respectively, then
\begin{enumerate}
\item Let $\ell$ be the number of edges on the walk $W(e,i)$.

\item Add an edge between $t_1$ and $t_2$ to $H$ with weight
$\frac{1}{\rho \ell}$.
\end{enumerate}
\end{enumerate}

\end{enumerate}

\end{enumerate}

\end{algbox}

\caption{Overall data structure, which is a
collection of $\beta$-shorted walks from
Definition~\ref{def:Walk} with $\beta = m^{-1/5}$,
reweighted according to Theorem~\ref{thm:SparsifySchur}.}
\label{fig:GlobalVariables}
\end{figure}

The performance of our data structures hinge upon the
properties of the random walks generated.
We start by formalizing such a structure involving a set of augmented
terminals, which we parameterize with a more general probability $\beta$.
\begin{definition}[$\beta$-shorted walks] \label{def:Walk}
Let $G$ be an unweighted, undirected multi-graph and $\beta \in (0,1)$ a parameter.
A collection of $\beta$-\emph{shorted walks} $W$ on $G$ is a set of random
walks created as follows:
\begin{enumerate}
\item Choose a subset of terminal vertices $T$, obtained by including
the endpoints of each edge independently, with probability at least $\beta$.
\item For each edge $e \in E$, generate $\rho$ walks from its endpoints
either until $O(\beta^{-2} \log^3{n})$ steps have been taken, or they reach $T$.
\end{enumerate}
\end{definition}

The main property of the collection $W$ is that its random walks are local.
That is, with high probability all walks in $W$ are short,
and only a small number of such walks pass through each vertex $u$, i.e., the expected load of $u$ with respect to $W$ is small.
These guarantees are summarized in the following theorem. Details on this behavior of random walks are deferred to Section~\ref{sec:propertiesRandomWalk}.

\begin{restatable}{theorem}{RandomWalkProperties}
\label{thm:RandomWalkProperties}
Let $G=(V,E)$ be any undirected multi-graph,
and $\beta \in (0,1)$ a parameter such that $\beta m = \Omega( \log{n})$. 
Any set of $\beta$-shorted walks $W$,
as described in Definition~\ref{def:Walk},
satisfies:
\begin{enumerate}
\item With high probability, any random walk in $W$ starting in a connected
component containing a vertex from $T$ terminates at a vertex in $T$.
\label{part:ReachT}
\item For any edge $e$, the expected load of $e$ with respect to $W$ is $O(\beta^{-2} \log^4{n} \epsilon^{-2})$.
\label{part:EdgeLoad}
\end{enumerate}
\end{restatable}

Note that Part~\ref{part:ReachT} is conditioned upon the connected component
having a vertex in $T$: this is necessary because walks stay inside a connected component.
However, this does not affect our queries:
our data-structure has an operation for making any vertex $u$ a
terminal, which we call during each query to ensure both $s$ and $t$
are terminal vertices.
Such an operation interacts well with Theorem~\ref{thm:RandomWalkProperties}
because it can only increase the probability of an edge's endpoints
being chosen.

We now have all the necessary tools to present our dynamic algorithm for maintaining the collection of walks $W$~(equivalently, the approximate Schur complement $H$). We start by analyzing the update operations. Upon insertion or deletion of an edge $e$ in the current graph $G$, the main idea is to regenerate all the walks that utilized $e$. This ensures that the collection of walks $W$ that we maintain produces a valid approximate Schur Complement $H$. Since we know that the length of these walks is $\tilde{O}(\beta^{-2})$~(Definition~\ref{def:Walk}), and
Theorem~\ref{thm:RandomWalkProperties} Part~\ref{part:EdgeLoad}
also limits the load per edge, using a rejection sampling technique,
we can regenerate these walks in $\Otil(\beta^{-4}\epsilon^{-2})$ time.

However, note that declaring $u$ to be a terminal forces us to truncate all the walks in $W$ at the first location they meet $u$.
Our key observation is that the cost of truncating these walks can be charged to the cost of constructing them during the pre-processing phase.
Thus it follows that we can declare any vertex a terminal in $O(1)$
amortized time. 
On the other hand, we need to avoid extending these truncated walk
when these query vertices are no longer needed in the terminals.
We address this by retraining the queried vertices in $T$, but periodically
rebuild the entire data structure (which which we resample the terminals
completely) to limit the growth in $|T|$.

The above discussion leads to the data-structure \textsc{DynamicSC$(G,T, \beta)$}: Given an undirected multi-graph $G=(V,E)$ and a subset of terminals $T$,
maintain the collection of $\beta$-shorted walks $W$,
and in turn the approximate Schur complement $H$~(as outlined in Figure~\ref{fig:GlobalVariables}) while supporting the following operations:
\begin{itemize}
\item \textsc{Initialize}$(G, T, \beta)$:
Construct the collection $W$ and the sparsifier $H$.
\item \textsc{Insert$(u,v)$}: Insert the edge $(u,v)$ with resistance $1$ in $G$.
\item \textsc{Delete$(u,v)$}: Delete the edge $(u,v)$ from $G$.
\item \textsc{AddTerminal$(u)$}: Add the vertex $u$ to the set of terminals in $T$.
\end{itemize}



We can now state the guarantees of the above data-structure using the defined operations. Specific implementation details are deferred to Section~\ref{sec:Dynamic}.

\begin{restatable}{lemma}{Dynamic}
\label{lem:Dynamic}
Given an undirected multi-graph $G=(V,E)$ a subset of
terminal vertices $T$, and a parameter $\beta \in (0,1)$ such that
$\beta m = \Omega( \log{n})$,
\textsc{DynamicSC}$(G,T, \beta)$ maintains the collection of $\beta$-shorted walks $W$, and in turn a graph $H$ that is with high probability a
sparsifier of $\textsc{SC}(G, T)$,
in a dynamic graph with at most $2m$ edges,
while supporting its operations in the following running times: 
\begin{enumerate}
\item \textsc{Initialize}$(G, T, \beta)$ in $O(m \beta^{-2} \log^5{n} \epsilon
^{-2})$ expected amortized time.
\label{case:initialize}
\item \textsc{Insert$(u,v)$} in
$O(\beta^{-4} \log^8{n} \epsilon^{-2})$ expected amortized time.
\label{case:Insert}
\item \textsc{Delete$(u,v)$} in
$O(\beta^{-4} \log^8{n} \epsilon^{-2})$ expected amortized time.
\label{case:Delete}
\item \textsc{AddTerminal$(u)$} in $O(1)$ amortized time.
\label{case:AddTerminal}
\end{enumerate}
Furthermore, each of these operations leads to an amortized number of changes
to $H$ bounded by the corresponding amortized costs.
\end{restatable}

Putting together the bounds in the above lemma proves our main result, i.e., Theorem~\ref{thm:Main}.

%

\begin{proof}[Proof of Theorem~\ref{thm:Main}]

We present our two-level data-structure for dynamically maintaining all-pair effective resistances. Specifically, we keep the terminal set $T$ of size $\Theta(\beta m)$. This entails maintaining
\begin{enumerate}
\item an approximate Schur complement $H$ of $G$~(Lemma~\ref{lem:Dynamic}),
\item a dynamic spectral sparsifier $\tilde{H}$ of $H$~(Lemma~\ref{lem:DynamicSpectralSparsifier}),
\end{enumerate}
and rebuilding our data-structure every $\beta m$ operations due to the
insertions into $T$ caused by handling queries.

We now describe the update and query operations. All updates in the graph are passed to the first data-structure~(which handles them by Lemma~\ref{lem:Dynamic} Parts~\ref{case:Insert} and~\ref{case:Delete}). Those updates in turn will trigger other updates in $H$, which are then handled by our second data-structure for $\tilde{H}$. Next, upon receiving a query about the $(s,t)$ effective resistance, we declare both $s$ and $t$ terminals~(by Lemma~\ref{lem:Dynamic} Part~\ref{case:AddTerminal}), which ensures that they are now contained in $\tilde{H}$. Finally, we compute the (approximate) effective resistance between $s$ and $t$ in the graph $\tilde{H}$ using Lemma~\ref{lemm:efficientEffectiveResistance}.

We next analyze the performance of our data-structure. Let us start with the pre-processing time. First, observe that the cost for constructing $H$ on a graph with $m$ edges is bounded by $\Otil(m \beta^{-2} \epsilon^{-2})$.
Next, since $H$ has $\Otil(m \epsilon^{-2})$ edges, constructing $\tilde{H}$ takes $\Otil(m \epsilon^{-4})$ time. Thus, the amortized time of our pre-processing is bounded by $\Otil(m\beta^{-2} \epsilon^{-4})$. 

We now analyze the update operations. By construction, note that a single update in $G$ triggers $\Otil(\beta^{-4} \epsilon^{-2})$ updates in $H$, and those updates can be handled in $O(\poly(\log n)\epsilon^{-2})$ time using the dynamic spectral sparsifier $\Htil$. Thus, we get that the expected amortized update time per insertion or deletion is $\Otil(\beta^{-4}\epsilon^{-4})$.

The cost of any $(s,t)$ query is dominated by (1) the cost of declaring $s$ and $t$ terminals and (2) the cost of computing the $(s,t)$ effective resistance to $\epsilon$ accuracy on the graph $\tilde{H}$. Since (1) can be performed in $O(1)$ time, we only need to analyze (2). We do so by first giving a bound on the size of $T$. To this end, note that each of the $m$ edges in the current graph adds two vertices to $T$ with probability $\beta$ independently.
By a Chernoff bound, the number of random augmentations added to
$T$ is at most $2\beta m$ with high probability.
In addition, since the data-structure in Lemma~\ref{lem:Dynamic} is re-built every $\beta m$ operations, the size of $T$ never exceeds $4\beta m$
with high probability.
The latter also bounds the size of $\Htil$ by $\Otil(\beta m\epsilon^{-2})$
and gives that the query cost is $\tilde{O}(\beta m \epsilon^{-4})$.

Finally, note that each rebuild can be performed in $\Otil(m \beta^{-2} \epsilon^{-4})$ amortized time.
Since we do rebuilds every $\beta m$ operations, this leads to an amortized cost of
\[
\frac{\Otil\left(m \beta^{-2} \epsilon^{-4}\right)}{\beta m}
=
\Otil\left(\beta^{-3} \epsilon^{-4}\right),
\]
which is dominated by the update time.

Combining the above bounds on the update and query time, we obtain the following trade-off
\[
	\tilde{O}\left((\beta m + \beta^{-4})\epsilon^{-4}\right),
\]
which is minimized when $\beta = m^{-1/5}$,
thus giving an expected amortized update and query time of 
\[
\tilde{O}\left(m^{4/5}\epsilon^{-4}\right).
\]


%
%
%
\end{proof}

\paragraph*{Future directions.} Our result raises several open questions for future works. Here we state those that we believe are closer to our results. (1) The most natural one is whether our update or query time bounds can be improved to $\poly(\log n, 1/\epsilon)$. It is also interesting to investigate whether our update times can be made deterministic and/or worst-case. (2) As effective resistances can be interpreted in terms of electrical flows, we can ask whether the values in these flows, such as the flow value
on edge $\ehat=(\uhat,\vhat)$ when sending $1$ unit of current from $\uhat$
to $\vhat$ can be dynamically maintained. (3) Finally, the use of Schur complement as vertex sparsifiers raises
the question of whether modifications of Schur complements can be used to maintain more
combinatorial problems.





\section{Properties of Random Walks}
\label{sec:propertiesRandomWalk}


In this section we give more details on the properties of $\beta$-shorted walks collection that our data-structure maintains.
Specifically, we prove Theorem~\ref{thm:RandomWalkProperties}:
\RandomWalkProperties*

We start with the first property, which claims that we traversed
sufficiently long to reach a vertex in $T$ with high probability.
For this, we need the following result by Barnes and Feige~\cite{BarnesF96}.

\begin{theorem}[\cite{BarnesF96}, Theorem 1.2]
\label{thm:ExpectedTime}
There is an absolute constant $c_{BF}$ such that for 
any undirected, unweighted, multi-graph $G$
with $n$ vertices and $m$ edges,
any vertex $u$ and any value $\mhat \leq m$,
the expected time for a random walk starting from $u$ to visit
at least $\mhat$ \emph{distinct} edges is at most $c_{BF} \mhat^2$.
\end{theorem}

The above theorem can be amplified into a with high probability
bound by repeating the walk $O(\log{n})$ times.

\begin{corollary} \label{cor:NumDistinctEdges}
In any undirected unweighted multi-graph $G$ with $m$ edges,
for any starting vertex $u$, any length $\ell$,
and a parameter $\delta \geq 1$,
a walk of length $c_{BF} \cdot \delta \cdot \ell \log n$ from $u$ visits
at least $\ell^{1/2}$ \emph{distinct} edges with probability at least $1 - n^{-\delta}$.
\end{corollary}

\begin{proof}
We can view each such walk as a concatenation of $\delta \log n$
sub-walks, each of length $c_{BF} \cdot \ell$.

We call a sub-walk \emph{good} if the number of distinct edges that
it visits is at least $\ell^{1/2}$.
Applying Markov's inequality to the result of Theorem~\ref{thm:ExpectedTime},
a walk takes more than $O(\ell)$ steps to visit $\ell^{1/2}$ distinct edges
with probability at most $1/2$.

This means that each subwalk fails to be good with probability at most $1/2$.
Thus, the probability that all subwalks fail to be good is at most
$2^{-\delta \log n} = n^{-\delta}$. The result then follows from an union bound over all starting vertices $u \in V$.
\end{proof}

This means that a walk of length $\Otil(\beta^{-2})$
is highly likely to visit at least $\beta^{-1} \log{n}$ distinct edges,
among which at least one should be added to $T$ with high probability.
If the connected component where the walk started in
has fewer than $\beta^{-1} \log{n}$ edges, we get that the walk should have
visited the entire component with high probability, and thus any (initial)
vertex in $T$ that belongs to that component.

\begin{proof}[Proof of Theorem~\ref{thm:RandomWalkProperties}
Part~\ref{part:ReachT}]
For any walk $w$, we define $V(w)$~(respectively, $E(w)$) to be the set of distinct vertices~(respectively, edges) that a walk $w$ visits. Consider a random walk $w$ that starts at $u$ of length
\[
\ell = c_{BF} \cdot \delta^3 \cdot  \beta^{-2} \log^{3} n
\]
where $\delta$ is a constant related to the success probability.

If the connected component containing the walk has fewer than
\[
\delta \cdot \beta^{-1} \cdot \log{n}
\]
vertices, then Corollary~\ref{cor:NumDistinctEdges}
gives that we have covered this entire component with high probability,
and the guarantee follows from the assumption that this component
contains a vertex of $T$.

Otherwise, we will show that $w$ reached enough edges for one of them
to be picked into $T$ with high probability.
The key observation is that because $w$ is generated independently from $T$, we can bound the probability of this walk not hitting $T$ by first generating $w$, and then $T$. Specifically, for any size threshold $z$, we have
\begin{align} \label{eqn: randomWalkProb}
 \prob{T, w}{V\left( w \right) \cap T  = \emptyset} & =
\prob{w, T}{V\left( w \right) \cap T = \emptyset}  \nonumber \\
& \leq
\prob{w}{\left| E\left( w \right) \right| \leq z}
+
\prob{w: \left| E\left( w \right) \right| \geq z}
{\prob{T}{V\left( w \right) \cap T = \emptyset}}.
\end{align}

By Corollary~\ref{cor:NumDistinctEdges} and the choice of $\ell$, if we set
\[
z = \delta \cdot \beta^{-1} \cdot \log{n},
\]
then the first term in Equation~(\ref{eqn: randomWalkProb})
is bounded by $n^{-\delta}$.
For bounding the second term, we can now focus on a particular
walk $\widehat{w}$
that visits at least $\delta \cdot \beta^{-1} \cdot \log{n}$
distinct edges, i.e.,
\[
\left|E\left(\widehat{w}\right)\right|
\ge
  \delta \cdot \beta^{-1} \log{n}.
\]

Recall that we independently added the end points of each
of these edges into $T$ with probability $\beta$.
If any of them is selected, we have a vertex that is both
in $V(w)$ and $T$.
Thus the probability that $T$ contains
no vertices from $V(\widehat{w})$ is at most
\[
\left(
  1 - \beta
\right)^{|E(\widehat{w})|}
\leq
\left(
  1 - \beta
\right)^{\delta \cdot \beta^{-1} \log{n}}
\leq
e^{- \delta \log n} 
\leq
n^{-\delta},
\]
which completes the proof.
\end{proof}

The bound on walk lengths also leads to a bound on the load
of an edge, which is Part~\ref{part:EdgeLoad}
of Theorem~\ref{thm:RandomWalkProperties}. We next show that this is the case. For the sake of simplicity, we will ignore the sampling overhead
$\rho = O(\log n/\epsilon^2)$ in our preliminary discussions. 

%

First, we observe that instead of terminating walks once they hit $T$ (as described in Figure~\ref{fig:GlobalVariables}), we can run all the walks from all edges up to $\ell = O(\beta^{-2} \log^{3} n)$ steps.

Note that the number of walks starting at each vertex $u$ is $\deg(u)$ because we are starting one random walk per endpoint of each edge. For each $u \in V$, we let $W(u)$ be the union over $\deg(u)$ random walks of length $\ell$ starting from $u$. Furthermore, recall that $W = \cup_{u} W(u)$ is the collection of $\beta$-shorted walks that our sparsification routine maintains. 

We want to obtain bounds on the load of any vertex $\uhat \in V$ (respectively, edge $\ehat \in E$) incurred by the random walks in $W$. For the purposes of the proof, it will be useful to introduce some random variables. The \emph{load} of $\uhat$ (respectively, $\ehat$), denoted by $N_{\uhat}$ (respectively, $N_{\ehat}$), is the number of walks that pass through vertex $\uhat$ (respectively, edge $\ehat$) from the random walks in $W$. For $t \geq 0$, let $X_u(t)$ be the set of vertices visited in a random walk starting at $u$ after $t$ steps. 

The first quantity we are interested in is the contribution of random walks from each $u \in V$ in the load of $\uhat$, which we denote by $Y_u(\uhat)$. Concretely, $Y_u(\uhat)$ is the total number of walks that pass through $\uhat$, from the random walks in $W_u$. Using the above random variables, we have that
\[
	Y_u(\uhat) = \sum_{0 \leq t \leq \ell} \deg(u) \cdot \boldone_{(\uhat \in X_u(t))}.
\] 

Now, observing that $N_{\uhat} = \sum_{u \in V} Y_u(\uhat)$, we can expand the expectation of $N_{\uhat}$ as follows
\begin{align*}
 \expec{}{N_{\uhat}} & = \sum_{u \in V} \expec{}{ Y_u(\uhat)} \\
 					 & = \sum_{u \in V} \sum_{0 \leq t \leq \ell} \deg(u) \cdot\prob{}{\uhat \in X_u(t)} \\
 					 & = \sum_{0 \leq t \leq \ell} \left( \sum_{u \in V} \deg(u) \cdot\prob{}{\uhat \in X_u(t)} \right) \numberthis \label{ExpectedLoadU}
.
\end{align*}

It turns out that that the term contained in the brackets of Equation~(\ref{ExpectedLoadU}) equals $\deg(\uhat)$. Formally, we have the following lemma.

\begin{lemma}
\label{lem:LULZ}
Let $G$ be an undirected, unweighted graph. For any vertex $\uhat \in V$ and any length $t \geq 0$, we have \normalfont
\[
\sum_{u \in V} \deg\left(u\right) \cdot \prob{}{\uhat \in X_u(t)}
= \deg\left( \uhat \right).
\]
\end{lemma}

To prove this, we use the reversibility of random walks, along with the
fact that the total probability over all edges of a walk starting at
$\ehat$ is $1$ at any time.
Below we verify this fact in a more principled manner.

\begin{proof}[Proof of Lemma~\ref{lem:LULZ}]
The proof is by induction on the length of the walks $t$. When $t = 0$, we have
\[
\prob {}{\uhat \in X_u(0)}
=
\begin{cases}
1 & \text{if $\uhat = u$},\\
0 & \text{otherwise},
\end{cases}
\]
which gives a total of $\deg(\uhat)$.

For the inductive case, assume the result is true for $t - 1$.
The probability of a walk reaching $\uhat$ after $t$ steps can then be written in terms of its location at time $t - 1$, the neighbor $\vhat$ of $\uhat$, as well as the probability of reaching there:
\[
\prob{}{\uhat \in X_u(t)}
=
\sum_{\vhat: (\uhat, \vhat) \in E}
\frac{1}{\deg\left( \vhat \right)}
\prob{} {\vhat \in X_u(t-1)}
.
\]
Substituting this into the summation to get
\[
\sum_{u \in V} \deg\left(u\right) \cdot \prob{}{\uhat \in X_u(t)}
=
\sum_{u \in V} \deg\left(u\right)
\sum_{\vhat: (\uhat, \vhat) \in E}
\frac{1}{\deg\left( \vhat \right)}
\prob{} {\vhat \in X_u(t-1)},
\]
which upon rearranging of the two summations gives:
\[
\sum_{\vhat: (\uhat,\vhat) \in E}
\frac{1}{\deg\left( \vhat \right)}
\left( 
\sum_{u \in V} \deg\left(u\right) \cdot
\prob{} {\vhat \in X_u(t-1)}\right).
\]
By the inductive hypothesis, the term contained in the bracket
is precisely $\deg(\vhat)$, which cancels with the division,
and leaves us with $\deg(\uhat)$.
Thus the inductive hypothesis holds for $t$ as well.
\end{proof}

\begin{proof}[Proof of Theorem~\ref{thm:RandomWalkProperties} Part~\ref{part:EdgeLoad}]

Plugging Lemma~\ref{lem:LULZ} into Equation~\ref{ExpectedLoadU} gives that
\[
	\expec{}{N_{\uhat}} \leq \deg(\uhat) \cdot \ell.
\]
Incorporating the sampling overhead $\rho$, which we initially ignored, we get
\begin{equation}
\expec{}{N_{\uhat}}
\leq \deg\left(\uhat\right) \cdot \ell \cdot \rho
\leq O\left(\deg\left(\uhat\right)
  \cdot  \beta^{-2} \log^{4} n \epsilon^{-2}\right),
\label{EqnloadVertex}
\end{equation} 
thus proving the desired bound on the load of $\uhat$.

To get the bound on the load of any edge $\ehat=(\uhat,\vhat)$, we use the fact that
\[
	\expec{}{N_{\ehat}} = \frac{1}{\deg\left( \uhat \right)}\expec{}{N_{\uhat}} + \frac{1}{\deg\left( \vhat \right)} \expec{}{N_{\vhat}}.
\]

Plugging the bound from Equation~(\ref{EqnloadVertex}) in the above equation, we get that
\[
   \expec{}{N_{\ehat}} \leq O\left(\beta^{-2} \log^{4} n \epsilon^{-2}\right),
\]
which proves the bound on the load of $\ehat$ and completes the proof.
\end{proof}

\section{Dynamic Schur Complement}
\label{sec:Dynamic}

In this section we show that the approximate Schur complement
given in Figure~\ref{fig:GlobalVariables} can be maintained dynamically. The primary difficulty here is dynamically maintaining the collection of $\beta$-shorted walks~(see Definition~\ref{def:Walk}). We next show how to do this efficiently and combine it with an amortized cost analysis to prove Lemma~\ref{lem:Dynamic}.
\Dynamic*

For our running time analysis, it is important to first note that each step in a random walk can be simulated in $O(1)$ time. This is due to the fact that we can sample an integer in $[0,n-1]$ by drawing $x\in[0,1]$ uniformly and taking $\lfloor xn \rfloor$.
Therefore, we will only need to consider the length of the walks in our runtime.
As we will later see, the initialization costs will then follow from our construction of the approximate Schur complement $H$ in Figure~\ref{fig:GlobalVariables}.

In addition, we note that there is always a one-to-one correspondence between the collection of $\beta$-shorted walks $W$ and our approximate Schur complement $H$.
Accordingly, our primary concern will be supporting the $\textsc{Insert}$, $\textsc{Delete}$, and $\textsc{AddTerminal}$ operations in the collection $W$.
However, as $W$ undergoes changes, we need to efficiently update the sparsifier $H$. To handle these updates, we would ideally have efficient access to which walks in $W$ are affected by the corresponding updates.





To achieve this, we index into walks that utilize a vertex
or an edge, and thus set up a reverse data structure pointing
from vertices and edges to the walks that contain them.
The following lemma says that we can modify this representation with minimal cost.
\begin{lemma}
\label{lem:ReversePointers}
For the collection of $\beta$-shorted walks $W$, let $W_v$ and $W_e$ be the specific walks of $W$ that contain vertex $v$ and edge $e$, respectively. 
We can maintain a data structure for $W$ such that for any vertex $v$ or edge $e$ it reports either

\begin{enumerate}
\item All walks in $W_v$ or $W_e$ in $O(|W_v|)$ or $O(|W_e|)$ time, respectively, or
\item The $i\textsuperscript{th}$ walk
(in order of time generated) of $W_v$ or $W_e$
in $O(\log{n})$ time, for any value $i$, 
\end{enumerate}
with an additional $O(\log{n})$ overhead for any changes made to $W$.
\end{lemma}

\begin{proof}
For every vertex~(respectively, edge), we can maintain a balanced binary search tree
consisting of all the walks that use it in time proportional
to the number of vertices~(respectively, edges) in the walks.
Supporting rank and select operations on such trees then gives the claimed bound.
\end{proof}

As a result, any update made to the collection of walks can be updated in the approximate Schur complement $H$ generated from these walks in $O(\log n )$ time.
Thus, we can fully devote our attention to supporting the $\textsc{Insert}$, $\textsc{Delete}$, and $\textsc{AddTerminal}$ operations in $W$.
The procedure $\textsc{AddTerminal}$ will be straightforward and its cost will be incorporated into the amortized analysis. The pseudocode for this operation is summarized in Figure~\ref{fig:AddTerminal}.

\begin{figure}

\begin{algbox}

$\textsc{AddTerminal}(u)$

\underline{Input}: vertex $u$ such that  $u \notin T$.

\begin{enumerate}
 
\item $T \leftarrow T \cup \{ u \}$.

\item Shorten all random walks to the first location they meet $u$.

\item Update the corresponding edges in $H$.

\end{enumerate}

\end{algbox}

\caption{Pseudocode for Adding a vertex to the set of terminals $T$}

\label{fig:AddTerminal}

\end{figure}

Procedures $\textsc{Insert}$ and $\textsc{Delete}$ are more involved.
The primary difficulty is that an edge update in the current graph changes the random walk distribution in the new graph and our walks may no longer be sampled according to this distribution. 
Recomputing each walk would be far too costly, so we instead observe that a random walk is a localized procedure. 
In particular, if some edge $(u,v)$ is inserted or deleted, the only changes in the random walk procedure occur when the walk visits vertex $u$ or $v$. 
This implies that all random walks in the original graph which did not visit vertex $u$ or $v$ have the same probability of occurring in the new graph with $(u,v)$ inserted or deleted.
We will then use certain properties of our collection of random walks proven in the previous section, along with a rejection sampling technique, to give stronger bounds on the number of walks we need to regenerate when the graph undergoes an edge update.


\subsection{Deletions}
\label{subsec:Deletions}

As mentioned above, due to the localized nature of random walks, if we delete an edge $(u,v)$, then all random walks in our collection that did not visit vertex $u$ or $v$ remain unaffected.
A simple update procedure would then regenerate all walks that visit vertex $u$ or $v$.
However, a consequence of Theorem~\ref{thm:RandomWalkProperties} Part~\ref{part:EdgeLoad} is that the expected number of walks visiting vertex $u$ is $O(\deg(u) \cdot \beta^{-2} \log^{4}{n} \epsilon^{-2})$,
which is too costly if $\deg(u)$ is large.
Ideally, we would then only deal with walks that use edge $(u,v)$, and we will next show that this is in fact that case.

To this end, note that the only random walks whose probability is affected are those that visit vertex $u$ or $v$.
Consider a random walk that visits vertex $u$ and let $\deg(u)$ be the degree of $u$ in the new graph.
For every remaining edge incident to $u$, the probability of it being used in the graph with $(u,v)$ deleted is $1/\deg(u)$.
Note that in the original graph they were used with probability $1/(\deg(u) + 1)$.
However, if we condition upon the random walk not using the edge $(u,v)$ in the original graph, then it is easy to see that any other edge incident to $u$ is chosen with probability $1/\deg(u)$, exactly as desired. Consequently, the only random walk probabilities that are affected are the ones that utilize edge $(u,v)$. In Figure~\ref{fig:DeleteEdge}, we give a routine which updates the walks that used the deleted edge $(u,v)$. The running time guarantees of our update algorithm are given in the following lemma.

%

\begin{figure}

\begin{algbox}

$\textsc{Delete}(u, v)$

\underline{Input}:
vertices $u$ and $v$ for which $(u,v)$ is an edge in $G$.

\underline{Output}:
an updated data-structure for $G \setminus \{(u,v)\}$.

\begin{enumerate}
\item Delete $(u,v)$ from $G$.
\item Delete all walks starting from $(u,v)$, as well as their
associated edges in $H$.
\item For each walk $w$ that uses the edge $(u,v)$:
\begin{enumerate}
\item Regenerate the walk from at the point it first reaches $u$
or $v$ using the remaining edges, until either $O(\beta^{-2}\log^3{n})$ steps
have been taken, or it reaches $T$.
\item Update in $H$ the edge corresponding to this walk.
\end{enumerate}
\end{enumerate}
\end{algbox}

\caption{Pseudocode for maintaining the collection of $\beta$-shorted walks $W$, and the corresponding graph $H$ after deleting edge $(u,v)$.}
\label{fig:DeleteEdge}
\end{figure}

\begin{lemma}
\label{lem:Deletion}
The operation $\textsc{Delete}(u, v)$ takes $O(\beta^{-4} \log^{8}{n} \epsilon^{-4})$ expected amortized time,
and updates {$O(\beta^{-2} \log^{4}n \epsilon^{-2})$} edges in $H$.
\end{lemma}

\begin{proof}

Theorem~\ref{thm:RandomWalkProperties} Part~\ref{part:EdgeLoad}
gives that the number of walks
that utilize $(u,v)$ is at most $O(\beta^{-2} \log^{4}n \epsilon^{-2})$, which in turn gives us the bound on the number of edges that need to be updated.

By the bound on the walk lengths in the collection of $\beta$-shorted
walks in Definition~\ref{def:Walk},
resampling each of these walks takes time $O(\beta^{-2} \log^3{n}\epsilon^{-2})$.
Adding an extra $O(\log{n})$ for translating between the collection of walks and $H$ by Lemma~\ref{lem:ReversePointers}, gives $O(\beta^{-4} \log^{8}n \epsilon^{-4})$ amortized expected time. 
\end{proof}

\subsection{Insertions}
\label{subsec:Insertions}

Handling insertions will be more involved because we now have to consider every random walk that visits $u$ or $v$. However, to bound the number of these walks that need to be regenerated, we can use rejection sampling.
This will utilize the fact that for any random walk that visits vertex $u$ or $v$, the probability that it uses edge $(u,v)$ is proportional to the degree of $u$ and $v$, respectively.
In particular, if we let $\deg(u)$ be the degree of $u$ in our new graph, then for each of our random walks that visit $u$ we instead use the edge $(u,v)$ with probability $1/\deg(u)$ and generate the remaining random walk from that point on.
Note that for any other edge incident to $u$, the probability of that edge being used in the original random walk was ${1}/{(\deg(u)-1)}$ and the probability that we keep the walk is $(\deg(u)-1)/\deg(u)$, whose product gives the desired probability of ${1}/(\deg(u))$.

However, if we run this sampling procedure for each walk incident to $u$, the expected number of walks sampled will be
$\Otil(\deg(u) \cdot \beta^{-2} \epsilon^{-2})$, which is too costly if the degree is large.
To address this, we implicitly run this sampling procedure on all walks by instead finding the instances in which we use $(u,v)$ in just $O(\log n)$ time.
More specifically, at each sample we use $(u,v)$ with probability ${1}/{\deg(u)}$, so the probability that we use $(u,v)$ for the first time in the $i$-th sample will be 

\[
\frac{1}{\deg(u)}\left(1 - \frac{1}{\deg(u)}\right)^{i-1}.
\]

In order to efficiently find the value $i$ at which we first sample $(u,v)$, we simply draw a uniformly random number $x \in [0,1]$.
Using geometric series properties, we can compute the probability that $i \leq j$ for any value $j$, and we binary search on $i \leq 2^j$ by increasing $j$ and checking if $x$ is above or below this value.
Our expected running time will then be $O(\log \deg(u)) \leq O(\log n)$.
Iterating this procedure efficiently finds all walks that use edge $(u,v)$.


Finally, note that we need to consider an additional case where the inserted edge causes both of its endpoints to be added to $T$. If this occurs, we simply truncate all walks to those vertices. The above discussion is summarized in the pseudocode given in Figure~\ref{fig:InsertEdge}

\begin{figure}

\begin{algbox}
$\textsc{Insert}(u, v)$

\underline{Input}:
vertices $u$ and $v$.

\underline{Output}:
an updated data-structure for $G \cup \{(u,v)\}$.

\begin{enumerate}
\item With probability $\beta$

\begin{enumerate}
\item $\textsc{AddTerminal}(u)$ and $\textsc{AddTerminal}(v)$

\end{enumerate}

\item Add edge $(u,v)$ to $G$.

\item Sample $\rho$ walks starting from $u$ and $v$
ending in $T$, add corresponding edges to $H$.

\item For each end point $\hat{u} \in \{u, v\}$

\begin{enumerate}

\item Let $d$ be the new degree of $\hat{u}$.

\item (implicitly) With probability $1/d$
for each occurrence of $\hat{u}$ in a random walk $w$:

\begin{enumerate}

\item Regenerate $w$ starting from that occurrence,
with $(u,v)$ as the edge taken.

\item Update in $H$ the edge corresponding to this walk.
\end{enumerate}

\end{enumerate}

\end{enumerate}
\end{algbox}

\caption{Pseudocode for maintaining the collection of $\beta$-shorted walks $W$, and the corresponding graph $H$ after inserting edge $(u,v)$.}
\label{fig:InsertEdge}
\end{figure}

%
\begin{lemma}
\label{lem:Insertion}
The operation
$\textsc{Insert}(u, v)$ takes $O(\beta^{-4} \log^{8}{n} \epsilon^{-4})$ expected amortized time, and updates $O(\beta^{-2} \log^4{n} \epsilon^{-2})$ edges in $H$.
\end{lemma}

\begin{proof}

For now, let us assume that the $\textsc{AddTerminal}$ operation takes $O(1)$ amortized time. We next bound the expected amortized time for the remainder of the procedure.

Theorem~\ref{thm:RandomWalkProperties} Part~\ref{part:EdgeLoad}
implies that the number of occurrences
of $u$ in all the walks is bounded by
\[
O \left(\deg\left( u \right) \cdot \beta^{-2} \log^{4}n \epsilon^{-2}\right).
\]
Since, by the discussion above, each of
these is updated with probability $1 / \deg(u)$,
the expected number of walks that we regenerate is
$O(\beta^{-2} \log^4{n} \epsilon^{-2})$, giving our bound on the number of edge updates. The total amortized cost then follows analogous to
the proof of Lemma~\ref{lem:Deletion}.
\end{proof}

\subsection{Amortized Analysis}

The overall bound requires amortizing the costs of shortening
the walks caused by $\textsc{AddTerminal}(u)$ to the cost of creating
these walks in the first place.
This can be handled using a standard amortized analysis.

\begin{proof}[Proof of Lemma~\ref{lem:Dynamic}]
	
	
We first examine the operation $\textsc{Initialize}(G,T,\beta)$. 
Forming $H$ requires generating $\rho = O(\log{n} \epsilon^{-2})$ walks
from each edge of $G$ up to a length of at most $O(\beta^{-2} \log^3{n})$,
which with the overhead of maintaining reverse pointers from
Lemma~\ref{lem:ReversePointers} gives a cost of
\[
O\left( m \beta^{-2} \log^5{n} \epsilon^{-2} \right).
\]

The expected amortized time of the operations $\textsc{Delete}(u,v)$ and $\textsc{Insert}(u,v)$ follow from Lemma~\ref{lem:Deletion} and~\ref{lem:Insertion}, respectively, where we note that Lemma~\ref{lem:Insertion} assumed that $\textsc{AddTerminal}(u)$ only required $O(1)$ amortized time.
It then remains to show that this is the case.

Adding a vertex to $T$ only shortens the existing walks, and 
Lemma~\ref{lem:ReversePointers} allows us to find such walks
in time proportional to the amount of edges deleted from the walk.
Since this walk needed to be generated in either the $\textsc{Initialize}$, $\textsc{Insert}$, or $\textsc{Delete}$, then the deletion of these edges will take equivalent time to generating them.
As a result, we can account for this amortized cost by just doubling the cost of $\textsc{Initialize}$, $\textsc{Insert}$, and $\textsc{Delete}$, which does not affect their asymptotic runtime.

\end{proof}

%

\section{Better Guarantees for $(s,t)$-Resistance on Simple Dense Graphs}
\label{sec:DependOnN}

In this section we discuss a different parameterization of our data structure where we restrict to \emph{simple} graphs and only maintain the effective resistance between a small number of \emph{fixed} vertex-pairs $(s_i,t_i)$. For the sake of exposition, we only consider the case where we want to maintain the effective resistance between a single $(s,t)$ pair. It is then easy to extend our data-structure to support up to $\Otil(n^{6/7})$ fixed vertex-pairs.


\begin{theorem} \label{thm:vertexBounds}
For any given error threshold $\epsilon > 0$, and a vertex-pair $(s,t)$, there is a data structure for maintaining a $n$-vertex simple graph $G=(V,E)$ while supporting edge insertions and deletions in $G$ as well as $(s,t)$-effective resistances queries in $\Otil(n^{6/7}\epsilon^{-4})$ expected update and query time.
\end{theorem}

In the above theorem, the improvement on the running time (for sufficiently dense graphs) comes from a result by Barnes and Feige~\cite{BarnesF96}, who give a bound on the number of distinct vertices visited in a random walk of certain length. In what follows, we will describe how to modify both the data-structure and the algorithm for maintaining dynamic Schur complements, which in turn will prove the theorem.

We start with the modification of the data-structure. In comparison to the data-structure in Figure~\ref{fig:GlobalVariables}, the key difference here is that we directly sample vertices with some probability and include them in the set of terminals $T$. This forces us to change the length of the random walks, as shown in Figure~\ref{fig:GlobalVariablesVertex}. However, we remark that the Theorem~\ref{thm:SparsifySchur} holds for arbitrary $T$, and thus it readily applies to our modification of the terminal set.

Now, for any parameter $\beta \in (0,1)$, similarly to Definition~\ref{def:Walk}, we can define the collection of $\beta$-shorted walks $W$ for the vertex version of our data-structure. Specifically, (1) we pick a subset of terminal vertices $T$, obtained by including each vertex independently, with probability at least $\beta$ and (2) the length of the random walks we generate is replaced by $O(\beta^{-3} \log^{4} n)$. 

\begin{figure}

\begin{algbox}

\begin{enumerate}

\item A subset of terminal vertices $T$, obtained by including each vertex independently, with probability at least $n^{-1/7}$.

\item A sampling overhead $\rho = O(\log n \varepsilon^{-2})$ (chosen according to Theorem~\ref{thm:SparsifySchur}).

\item Graph $H$ created by
repeating the following procedure
for each edge $e = (u,v)$,

\begin{enumerate}
\item For $i=1,\ldots,\rho$,

\begin{enumerate}
\item Generate random walks $W(e, i)$
from $u$ and $v$ until either $O(n^{3/7}\log^3{n})$ steps
have been taken, or they reach $T$.

\item If both walks reach $T$ at $t_1$ and $t_2$ respectively, then
\begin{enumerate}
\item Let $\ell$ be the number of edges on the walk $W(e,i)$.

\item Add an edge between $t_1$ and $t_2$ to $H$ with weight
$\frac{1}{\rho \ell}$.
\end{enumerate}
\end{enumerate}

\end{enumerate}

\end{enumerate}

\end{algbox}

\caption{Overall data structure, which is a
collection of $\beta$-shorted walks from with $\beta = n^{-1/7}$,
reweighted according to Theorem~\ref{thm:SparsifySchur}.}
\label{fig:GlobalVariablesVertex}
\end{figure} 
 We next review the result by Barnes and Feige~\cite{BarnesF96}.

\begin{lemma}[\cite{BarnesF96}, Theorem 1.1]
\label{lem:ExpectedTimeVertices}
There is an absolute constant $c_{BF}$ such that for any undirected, unweighted, simple $G$ with $n$ vertices, any vertex $u$ and any value $\nhat \leq n$, the expected time for a random walk starting from $u$ to visit at least $\nhat$ \emph{distinct} vertices is at most $c_{BF}\nhat^{3}$.
\end{lemma}

Now, using the above lemma and following essentially the same reasoning as in Section~\ref{sec:propertiesRandomWalk}, we get the analogue of Theorem~\ref{thm:RandomWalkProperties}.

\begin{theorem}
Let $G=(V,E)$ be any undirected simple graph,
and $\beta \in (0,1)$ a parameter such that $\beta n = \Omega( \log{n})$. 
Any set of $\beta$-shorted walks $W$,
as described above,
satisfies:
\begin{enumerate}
\item With high probability, any random walk in $W$ starting in a connected
component containing a vertex from $T$ terminates at a vertex in $T$.
\item For any edge $e$, the expected load of $e$ with respect to $W$ is $O(\beta^{-3} \log^5{n} \epsilon^{-2})$.
\end{enumerate}
\end{theorem}

Following the ideas we presented in Section~\ref{sec:Overview}, we can use the above theorem to construct the data-structure that maintains a dynamic Schur complement, i.e., \textsc{DynamicSC}$(G,T,\beta)$. However, one difference here is that the terminal additions are not supported, and the update times are no longer amortized. We implement the insertions and deletions of edges using Lemmas~\ref{lem:Deletion} and \ref{lem:Insertion}. Note that because we randomly pick vertex subsets, we do not need to regenerate augmentations to $T$, i.e., Line 1 of Figure~\ref{fig:InsertEdge} is no longer useful in our \textsc{Insert} routine. The guarantees of these modifications are summarized in the lemma below.


\begin{lemma}
\label{lem:Dynamic2}
Given an undirected multi-graph $G=(V,E)$ a subset of
terminal vertices $T$, and a parameter $\beta$ such that
$\beta n = \Omega( \log{n})$,
\textsc{DynamicSC}$(G,T, \beta)$ maintains the collection of $\beta$-shorted walks $W$, and in turn a graph $H$ that is with high probability a
sparsifier of $\textsc{SC}(G, T)$, while supporting its operations in the following running times: 
\begin{enumerate}
\item \textsc{Initialize$(G,T,\beta)$} in $O(m\beta^{-3} \log^{5} n \epsilon^{-2})$ expected time.
\label{case:Initialize2}
\item \textsc{Insert$(u,v)$} in
$O(\beta^{-6} \log^9{n} \epsilon^{-2})$ expected  time.
\label{case:Insert2}
\item \textsc{Delete$(u,v)$} in
$O(\beta^{-6} \log^9{n} \epsilon^{-2})$ expected time.
\label{case:Delete2}
\end{enumerate}
Furthermore, each of these operations leads to a number of changes
to $H$ bounded by the corresponding costs.
\end{lemma}



Putting together the bounds in the above lemma proves our vertex based bounds, i.e., Theorem~\ref{thm:vertexBounds}.

\begin{proof}[Proof of Theorem~\ref{thm:vertexBounds}]

We present our two-level data-structure for dynamically maintaining $(s,t)$-effective resistances. Specifically, we include both $s$ and $t$ to $T$, and keep the terminal set $T$ of size $\Theta(\beta n)$. This entails maintaining
\begin{enumerate}
\item an approximate Schur complement $H$ of $G$~(Lemma~\ref{lem:Dynamic2}),
\item a dynamic spectral sparsifier $\tilde{H}$ of $H$~(Lemma~\ref{lem:DynamicSpectralSparsifier}).
\end{enumerate}

We now describe the update and query operations. Specifically, whenever an edge insertion or deletion is performed in the current graph, we pass the corresponding update to the first data-structure~(which handles this by Lemma~\ref{lem:Dynamic2} Parts~\ref{case:Insert2} and~\ref{case:Delete2}). This update in turn will trigger other updates in $H$, which are then handled by our second data-structure for $\tilde{H}$. Next, upon receiving a query about the $(s,t)$ effective resistance, we compute the (approximate) effective resistance between $s$ and $t$ in the graph $\tilde{H}$ using Lemma~\ref{lemm:efficientEffectiveResistance}.

Similarly to the proof of Theorem~\ref{thm:Main}, we can show that the pre-processing time is $\Otil(m\beta^{-3} \epsilon^{-4})$, the expected update time is $\Otil(\beta^{-6} \epsilon^{-4})$, and the query time is $\Otil(\beta n \epsilon^{-4})$.

Combining the bounds on the update and query time, we obtain the following trade-off
\[
	\tilde{O}\left((\beta n+ \beta^{-6})\epsilon^{-4}\right),
\]
which is minimized when $\beta = n^{-1/7}$,
thus giving an expected update and query time of 
\[
\tilde{O}\left(n^{6/7}\epsilon^{-4}\right).
\]
\end{proof}


\subsection*{Acknowledgements}
We thank Daniel D. Sleator for helpful comments on an earlier draft of the manuscript.


\bibliographystyle{alpha}
\bibliography{ref}

\begin{appendix}

\section{Schur Complement Sparsifier from Sum of Random Walks}
\label{sec:SchurComplement}

In this section, we prove Theorem~\ref{thm:SparsifySchur},
which states that sampling random walks generates sparsifiers of Schur complements:
\SparsifySchur*

Note that this rescaling by $1 / \rho \ell$ is quite natural:
in the degenerate case where $T=V$,
this routine generates $\rho$ copies of each edge,
which then need to be rescaled by $1 / \rho$ to ensure approximation
to the original graph. 

Similar to other randomized graph sparsification
algorithms~\cite{SpielmanS08:journal,KoutisLP15,AbrahamDKKP16,DurfeePPR17,JindalKPS17},
our sampling scheme directly interacts with Chernoff bounds. Our random matrices are `groups' of edges related to random walks
starting from the edge $e$. We will utilize Theorem 1.1 due to~\cite{Tropp12}, which we paraphrase in our notion of approximations.

\begin{theorem} 
Let $\normalfont \XX_{1}, \XX_{2} \ldots \XX_{k}$ be a set of random matrices satisfying the following properties:
\begin{enumerate}
\item Their expected sum is a projection operator onto some subspace, i.e., 
\[
\normalfont \sum_{i} \expec{}{\XX_i} = \PPi.
\]
\item For each $\XX_{i}$, its entire support satisfies:
\[
0 \preceq \XX_{i} \preceq \frac{\epsilon^2}{O\left( \log{n} \right)} \II.
\]
\end{enumerate}
Then, with high probability, we have
\[
\sum_{i} \XX_{i} \approx_{\epsilon} \PPi.
\]
\end{theorem}

Re-normalizations of these bounds similar to the work of~\cite{SpielmanS08:journal}
give the following graph theoretic interpretation of the theorem above.
\begin{corollary}
\label{cor:Sparsify}
Let $E_1 \ldots E_k$ be distributions over random edges satisfying the following properties:
\begin{enumerate}
\item Their expectation sums to the graph $G$, i.e.,
\[
\sum_{i} \expec{}{E_i} = G.
\]
\item For each $E_{i}$, any edge in its support has
low leverage score in $G$, i.e., 
\[
\normalfont \ww_{e} \er^{E_i} \left( e \right)
\leq
\frac{\epsilon^2}{O\left( \log{n} \right)}.
\]
\end{enumerate}
Then, with high probability, we have
\[
\sum_{i} \LL_{E_{i}} \approx_{\epsilon} \LL_{G}. 
\]
\end{corollary}

To fit the sampling scheme outlined in Theorem~\ref{thm:SparsifySchur}
into the requirements of Corollary~\ref{cor:Sparsify},
we need (1) a specific interpretation of Schur complements in terms of walks, and (2) a bound on the effective resistances between two vertices at a given distance.

Given a walk $w = u_0,\ldots,u_{\ell}$ of length $\ell$ in $G$ with a subset a vertices $T$, we say that $w$ is a \emph{terminal-free} walk iff $u_0,u_{\ell} \in T$ and $u_1,\ldots,u_{\ell-1} \in V \setminus T$.

\begin{fact}[\cite{DurfeePPR17}, Lemma 5.4]\label{fact:WalkDecomposition}
\footnote{We state the lemma for unit weighted graphs. The version cited is~\url{https://arxiv.org/pdf/1705.00985v1.pdf}
There may be updates to this arXiv manuscript in the near future.}
For any undirected, unweighted graph $G$
and any subset of vertices $T \subseteq V$,
the Schur complement $\SC(G,T)$ is given as an union over all multi-edges corresponding to terminal-free walks $u_{0}, \ldots ,u_{\ell}$ with weight
\[
\prod_{i = 1}^{\ell - 1} \frac{1}{\deg\left( u_{i} \right)}.
\]
\end{fact}
The fact below follows by repeatedly applying the triangle inequality of the effective resistances between two vertices.
\begin{fact}
\label{fact:ERBound}
In an unweighted undirected graph $G$, the effective resistance
between two vertices that are distance $\ell$ apart is at most $\ell$.
\end{fact}

Combining the above results gives the guarantees of our sparsification routine.

\begin{proof}[Proof of Theorem~\ref{thm:SparsifySchur}]
For every edge $e \in E$, let $W_e$ be the random graph corresponding the the terminal-free random walk that started at edge $e$. Define $H = \rho \cdot \sum_{e} W_e$ to be the output graph by our sparsification routine, where $\rho= O(\log n \epsilon^{-2})$ is the sampling overhead. To prove that $\LL_H \approx_{\epsilon} \SC(G,T)$ with high probability, we need to show that (1) $\expec{}{H} = \SC(G, T)$ and (2) for any edge $f$ in $W_e$, its leverage score $\ww_f \er^{W_e}( f )$  is at most $\leq \epsilon^{2}/ \log n$ (by Corollary~\ref{cor:Sparsify}). Note that (2) immediately follows from the effective resistance bound of Fact~\ref{fact:ERBound} and the choice of $\rho =O(\log n / \epsilon^2)$. We next show (1).

To this end, we start by describing the decomposition of $\SC(G, T)$ into
random multi-edges, which correspond to random terminal-free walks in Fact~\ref{fact:WalkDecomposition}. The main idea is to sub-divide each walk 
$u_0 \ldots u_{\ell}$ of length $\ell$ in $G$ into $\ell$ walks of the same length, each starting at one of the $\ell$ edges on the walk, and each having weight
\[
\frac{1}{\ell}
\cdot
\prod_{i = 1}^{\ell - 1} \frac{1}{\deg\left( u_{i} \right)}.
\]
By construction of our sparsification routine, note that every random graph $W_e$ is a distribution over walks $u_0 \ldots u_{\ell}$, each picked with probability
\[
\prod_{i = 1}^{\ell - 1} \frac{1}{\deg\left( u_{i} \right)}.
\]
Thus, to retain expectation, when such a walk is picked, our routine correctly adds it to $H$ with weight $1/(\rho \ell).$ 

Formally, we get the following chain of equalities
\begin{align*}
 \expec{}{H} & = \rho \cdot \sum_{e} \expec {}{W_e}  \\
             & = \rho \cdot \sum_{e} \sum_{ w = u_0, u_1 \ldots u_{\ell\left( w \right)} : w \ni e} \frac{1}{\rho \ell\left( w \right)} \cdot  \prod_{i=1}^{\ell(w)-1} \frac{1}{\deg\left( u_{i} \right)}  \\
             & = \sum_{w = u_0, u_1 \ldots u_{\ell\left( w \right)}} \sum_{e : e \in w} \frac{1}{\ell\left(w \right)} \cdot \prod_{i=1}^{\ell(w)-1} \frac{1}{\deg\left( u_{i} \right)}  \\
             & = \sum_{w = u_0, u_1 \ldots u_{\ell\left( w \right)}} \prod_{i=1}^{\ell\left( w \right)-1}
                \frac{1}{\deg\left( u_{i} \right)}  \\ 
             &= \SC(G,T).
\end{align*}
\end{proof}

%

\section{A Unified View of Flows and Paths}
\label{sec:Flows}


We provide a brief overview of numerical formulations of flows
that capture combinatorial problems including shortest paths, maximum flows, and effective resistances.
This view is well known in the literature of using continuous methods
for combinatorial optimization problems~\cite{ChristianoKMST11,Madry11:thesis}.

For an orientation of edges of a graph $G$ with
$n$ vertices and $m$ edges, we can define the
edge-vertex incidence matrix $\BB \in \mathbb{R}^{m \times n}$ as:
\[
\BB_{e, u}
:=
\begin{cases}
1 & \text{if $e$ is the head of $u$},\\
-1 & \text{if $e$ is the tail of $u$},\\
0 & \text{otherwise}.
\end{cases}
\]

Then a flow from $s$ to $t$ is a vector $\ff$ on edges such that
\[
\BB^{T} \ff = \cchi_{st},
\]
where $\cchi_{st}$ is the indicator vector with $1$ at $t$,
$-1$ at $s$, and $0$ everywhere else.

Furthermore, for any $p \geq 1$, we can define the $p$-norm of a flow $\ff$ via
\[
\norm{\ff}_{p}
:=
\left( \sum_{e} \abs{\ff_{e}}^{p} \right)^{1/p}.
\]

Shortest paths, maximum flows, and electrical flows~(effective resistances)
on undirected graphs are all instances of the following optimization problem:
\begin{align*}
\min \qquad & \norm{\ff}_p\\
\text{subject to:} \qquad & \BB^{T} \ff = \cchi_{st}.
\end{align*}
Specifically, we distinguish the following cases:
\begin{enumerate}
\item When $p = 1$, we get the shortest path problem between $s$ and $t$. Replacing $\cchi_{st}$ with a more general demand vector $\dd$, we get
the transshipment problem~\cite{Sherman17}.
\item When $p = \infty$, we get the problem of minimizing congestion,
which is equivalent to routing the maximum amount of flow from $s$ to $t$
subject to at most $1$ unit per edge, or in turn the undirected 
max-flow/min-cut problem.
\item When $p = 2$, we get the $s-t$ electrical flow problem.
Here, since $\norm{\ff}_2^{2}$ is differentiable, we have
\[
\ff^{T} \Delta = 0
\]
for any `change' that is a circulation, i.e, $\BB^T \Delta = 0$.
The properties of column or row spaces then imply that
\[
\ff = \BB \pphi,
\]
for a voltage vector $\pphi$, for which we can then solve to get
\[
\BB^{T} \BB \pphi = \cchi_{s, t},
\]
or $\pphi = \LL^{\dag} \cchi_{s, t}$, since $\LL = \BB^{T} \BB$.
The energy of the resulting flow $\ff$ is then:
\[
\norm{\ff}_2^2
=
\norm{\BB \LL^{\dag} \cchi_{s, t}}_2^2
= \cchi_{s, t}^{T} \LL^{\dag} \cchi_{s, t}, 
\]
which is exactly the definition of $s$-$t$ effective resistance
from Section~\ref{sec:Preliminaries}.
\end{enumerate}

\end{appendix}  

\end{document}